%% file: paper.tex
\let\oldvec\vec

\documentclass[runningheads,orivec]{llncs}
\usepackage{amsbsy}
\usepackage{macros}
\usepackage[numbers,sort,sectionbib]{natbib}
\hypersetup{draft=true}
\tightams{2}

\let\vec\oldvec

\begin{document}

\begin{frontmatter}

\title{%
(Towards a) Statistical Probabilistic Lazy Lambda Calculus%
  }

\author{Radha Jagadeesan}

\institute{%
  School of CTI, DePaul University, Chicago, IL 60604, USA \\
  \email{rjagadee@depaul.edu}
}

\maketitle

\begin{abstract}
We study  the desiderata on a model for statistical probabilistic programming languages.  We argue that they can be met by a combination of traditional tools, namely open bisimulation and probabilistic simulation.  
\end{abstract}
\end{frontmatter}

\pagenumbering{arabic}

\input{intro}

 \input{language}

 \input{domain}

\input{end}

\bibliographystyle{ACM-Reference-Format}
\bibliography{bib}

\appendix
\input{appendix}

\end{document}

%% file: intro.tex
\section{Introduction}
The thesis of this article is that open bisimulation is a useful methodology to describe the semantics of statistical probabilistic languages (henceforth \statprob).  

~\citet{vandemeent2018introduction} provides a textbook introduction to these languages.  To a base programming
language, \eg, with state, higher order functions and recursion, add two key features to get a statistical probabilistic language.
\begin{description}
\item[Forward simulation: ]   The ability to sample/draw values at random from (perhaps continuous) distributions, to describe forward  stochastic simulation.
\item[Inference: ]  The ability to condition the posterior distributions of values of variables via observations, and compute them via inference.
\end{description}
~\cite{vandemeent2018introduction} describes \statprob-languages as ``performing Bayesian inference using the tools of
computer science: programming language for model denotation and statistical
inference algorithms for computing the conditional distribution
of program inputs that could have given rise to the observed program
output.''.  ~\citet{10.1145/3290349} illustrates these key features with the following example: 
\begin{verbatim}
let a = normal(0,2) in
  score(normal(1.1 | a*1, 0.25))
  score(normal(1.9 | a*2, 0.25))
  score(normal(2.7 | a*3, 0.25))
in a
\end{verbatim}
The program postulates a linear function $a \times x$ with slope $a$. Three noisy measurements $1.1,1.9,2.7$ at $x=1,2,3$  respectively are postulated to have noise modeled as a normal distribution with standard deviation $0.25$.  The goal
is to find a posterior distribution on the slope $a$.  Operationally, a sampler draw from the prior, which is a normal with  mean $0$ and standard deviation $2$.  The posterior distribution is constrained using the resulting samples of the three data points.  While the forward simulation yields the un-normalised posterior distribution, an inference algorithm is used to approximate its normalization and compute the posterior distribution on $a$.  Arguably, the evaluators that implement such conditioning are the proximate cause  for the recent explosion in interest in this area. 

In the rest of this introduction, we provide an overview of open bisimulation and the rationale for its use in the current setting.

\subsection{Open bisimulation}
Consider the study of applicative bisimulation in the pure untyped lazy lambda calculus by~\citet{10.5555/119830.119834},~\citet{ ABRAMSKY1993159} and~\citet{Ong88}.  In particular, recall ~\citet{Gordon95}'s LTS approach to applicative bisimulation.
\begin{description}
\item[Internal: ] A non-value term $\atrm$ has a $\tau$ transition to $\atrm'$ if
  $\atrm$ reduces in one step to $\atrm'$.
\item[Applicative test: ] $\aval = \abs{\avar}{\atrm}$ has a transition
  labeled $\aval'$ to the application $\aval \A \aval'$.
\end{description}
Two terms are bisimilar if the associated transition systems are bisimilar, \ie, if their convergence properties agree and each
applicative test yields bisimilar terms.      

In this article, we build on~\citet{LASSEN1999346}'s presentation of ~\citet{SANGIORGI1994120}'s treatment.  We follow the ``open bisimulation'' presentation of Lassen's later work~\citep{DBLP:conf/lics/Lassen05,DBLP:conf/lics/LassenL08,DBLP:conf/birthday/StovringL09} in a form that validates $\eta$-expansion.   In contrast to applicative bisimulation, the applicative tests are restricted to (perhaps new) variable names $\asym$:
\begin{description}
\item[Applicative test: ] $\aval=\abs{\avar}{\atrm}$ has a transition
  labeled $\ltsapp{\asym}$ to the application $\aval \A \asym$.
\end{description}
``Open'' application results in a new category of normal forms of the form $\asym \A \atrms$, for a vector of terms $\atrms$, forcing the consideration of ``external call'' tests:
\begin{description}
\item[External call: ] A value $\asym \A \atrms$ has a transition
  labeled $\ltsfuncall{\asym}{\PBR{i,n}}$ to $\atrms\PBR{i}$ (the $i$'th term in $\atrms$), where $|\atrms| = n$.
\end{description}
~\citet{LASSEN1999346} shows that Levy-Longo tree equivalence is the biggest bisimulation for such a transition system.    In subsequent papers beginning with\citep{DBLP:conf/lics/Lassen05}, Lassen and his coauthors studies a full suite of programming language features, including control~\cite{DBLP:conf/lics/Lassen06}, and state~\citep{DBLP:conf/birthday/StovringL09}.~\citet{DBLP:journals/taosd/JagadeesanPR09} adapt this approach to aspect-oriented programming languages.    The open-bisimulation approach has also been used to study issues related to types and parametricity, \eg, ~\citet{DBLP:conf/lics/LassenL08} provide a coinduction principle to simplify proofs about existential types and~\citet{DBLP:conf/fossacs/JaberT18} show that Strachey parametricity implies Reynolds parametricity. 

Why does this approach work so well?  ~\citet{DBLP:conf/csl/LevyS14} provides a hint by demonstrating a tight correspondence between open bisimulation and game semantics, ~\citet{AbramskyJM00},~\citet{HylandO00}, and~\citet{DBLP:conf/lfcs/Nickau94}, albeit for a small calculus with a limited type theory.  This has led Paul Levy to coin the evocative (but arguably unnattractive?!) term ``Operational Game Semantics'' to describe the approaches based on open bisimulation. In this article, we do not further the study needed to resolve the suspicion of  a more pervasive and general connection between game semantics and open bisimulation.    Rather, what is crucial is the commonality between the semantics of programs in these two approaches, namely:
\begin{center}
{\em The semantics of a program is a set of traces}.
\end{center}
Remarkably, there is no explicit mention of any higher order structure, such as abstraction and higher order functionals.  

On the games side, this is particularly true of the~\citet{AbramskyJM00} style of games; the~\citet{HylandO00} games and~\citet{DBLP:conf/lfcs/Nickau94} games also appeal to a ``pointer'' structure that is not intrinsic to a linear trace.  The AJM games inherit this feature from the Geometry of Interaction of~\citet{GIRARD1989221}.   In this style of model, a single top level, first order,  fixpoint iterator serves as the control engine driving the execution.  The complications and subtlelty of higher order control flow is handled purely by dataflow over structured tokens, an idea made precise in~\citet{DBLP:journals/iandc/AbramskyJ94}.

On the open bisimulation side, the transition system identifies the (set of) paths through the normal form of the program as the invariant of the bisimulation class of the program;  a familiar perspective of GOI models as seen in~\citet{151631}.   The only residue of higher order computation in this first-order world view is perhaps the partition of the labels of the transition system into moves associated with the program and the environment.  

This seemingly abstruse technical point acquires particular relevance when considering the semantics of statistical programming languages.

\subsection{Semantics of statistical probabilistic languages}
There are two serious impediments to a systematically structured study of statistical probabilistic languages.  

First, the semantics of  a simple (pseudo)-random value generators  is already troublesome.  Standard probability theory does not interact well with higher order functions, since the category of measurable spaces is not cartesian closed.  From the domain theoretic perspective,~\citet{10.5555/77350.77370} provide a construction of a probablistic powerdomain, and show that the probabilistic powerdomain of a continuous domain is again continuous.  The issues with reconciling the probabilistic powerdomain and function spaces are well known, and remain unresolved to date; see~\citet{JUNG199870}.   This leads~\citet{HOSY17} to identify the mismatch with probabilistic  programming languages where ``programs may use both higher-order functions and
continuous distributions, or even define a probability distribution on functions.''

Second, the accurate modeling of conditional probabilities is subtle in the presence of possibly infinite computations.  This computation requires normalization that is not monotone wrt the order on (sub)-probability distributions; \eg, consider the following weighted sums on  distinct constants $\tru,\fls$:
\[ \PBR{\pterm{0.2}{\tru}, \pterm{0.2}{\fls}} \leq \PBR{\pterm{0.2}{\tru}, \pterm{0.3}{\fls}} \]
whereas after normalization, we get:
\[ \PBR{\pterm{0.5}{\tru}, \pterm{0.5}{\fls}} \not\leq \PBR{\pterm{0.4}{\tru}, \pterm{0.6}{\fls}} \]
Thus, the usual view of infinite computations as a supremum of an increasing chain of finite approximants needs to be revisited; we are unable to rely purely on the order structure to compute the probability.    This issue may not be fatal in the setting of purely countable probabilities; however, the setting of continuous probabilities perforce requires a notion of limits that is different from the one provided by the order theory of simulation.  

\paragraph*{Our approach. } 
We address these issues as follows.  Our calculus of choice is an untyped lambda calculus with binary choice $\Lambda$, with a lazy reduction scheme to weak head normal forms\footnote{The choice of an untyped calculus is merely to reduce the syntactic overhead in the treatment. }.
 
We define open probabilistic simulation, $\simreln$, by a direct combination of two traditional ingredients.  
\begin{itemize}
\item  Open (bi)simulation for lazy lambda calculi, adapted to account for discrete distributions.  

Probability impacts the notion of convergence.  Let $\ycomb$ be a fixpoint combinator , $\Id$ be $\abs{\avar}{\avar}$, and let $\atrm = \abs{\avar}{ \CBR{\pterm{0.5}{\Id}, \pterm{0.5}{\avar}}}$.  Then, $\ycomb \A \atrm$ converges with probability $1$ to $\Id$, but the recognition of this fact {\em requires} the complete infinite unfolding of the computation.  

\item  Probabilistic (weak) simulation for labelled transition systems~\citep{LS}.

This definition incorporates a ``splitting lemma'' characteristic of probabilities in various treatments in the literature, recognized by~\citet{10.1109/LICS.2007.15} to be the Kleisli construction on the  probabilistic power domain.   Probabilistic determinacy is crucial, and ensures that the closure ordinal of the simulation functional is $\omega$.  
\end{itemize}
We demonstrate that simulation is a congruence.  In the context of the simple discrete probability embodied in this calculus, this development already shows how the first-order treatment of open bisimulation allows us to sidestep the troublesome interaction of probability and function spaces.

Next, we develop an order theory of simulation.  In effect, this study is an examination of $D = {\cal P}_{\text{prob}}(D \Rightarrow\ D)$, except for the unfortunate accident that we do not know how to construct such a $D$!   We adapt the finite Levy-Longo trees of~\citet{Ong88} to account for discrete probabilistic sums, and show that it constitutes  a basis for $\PBR{\Lambda,\simreln}$.    This structure suffices to define ``the'' Lawson topology~\citep{contlat} on $\PBR{\Lambda,\simreln}$, and define a metric completion that yields a compact Polish space $\close{\lquot}$.   This has two consequences:
\begin{itemize}
\item Since an omega $\simreln$-chain is a convergent sequence in $\close{\lquot}$, this completion makes $\close{\lquot}$ into a continuous dcpo.  Since $\PBR{\Lambda,\simreln}$ already contains finite valuations, we deduce that $\close{\lquot}$ is rich enough to contain continuous probability distributions, along with an approximation theory by finite valuations.  

\item ~\citet{DBLP:conf/fossacs/BreugelMOW03} show that the Lawson topology on the probabilistic powerdomain of a continuous dcpo $D$ coincides with the weak topology on the probabilistic measures of the Lawson topology of $D$. Convergence in $\close{\lquot}$ thus provides a meaningful definition of the convergence of probability measures.  In the specific example of normalization,  this suggests a compelling operational intuition to describe the results of normalizing the computation represented by a term $\atrmm$. Consider the net of normalized finite evolutions of $\atrmm$; the result of normalization exists if the net is convergent.     
\end{itemize}

\paragraph{Other approaches. }
Rather than consider a separate probabilistic powerdomain construction, there is a line of work that considers types that are already closed under probability.  In this spirit,~\citet{10.1145/507382.507385} describe a semantics for a probabilistic extension of Idealized Algol; the key enhancement to traditional game semantics is to move from deterministic strategies to probabilistically deterministic strategies.  ~\citet{DANOS2011966} describe probabilistic coherence spaces where types are interpreted  as convex sets and programs as power series.~\citet{10.1145/2535838.2535865} establish a remarkable equational (but not inequational!) full abstraction result for this semantics.~\citet{Goubault-Larrecq2019} motivates why such a result holds by showing how to define a ``poor mans parallel or'' in a probabilistic language.  

~\citet{10.1145/3350618} use environmental bisimulation in a thorough treatment of probabilistic languages.  Environmental bisimulation tracks an observer's knowledge about values computed during the bisimulation game with aim of simplifying the delicate proofs of congruence in applicative bisimulation.    We gratefully acknowledge their impact on the formalization of the operational aspects of the calculus.  

Both the above lines of research resolve the tension between the probabilistic powerdomain and function spaces. However, they do not account for the normalization of probability distributions.    It is plausible that our techniques that study the order theory of simulation can be adapted to these settings.  

A second line of work is on Quasi Borel Spaces~\citep{10.1145/3290349,HOSY17}.   In our reading of this approach, it diagnoses the problems of modeling statistical probabilistic languages as arising from the interplay of approximating data and approximating probabilities.  So, the radical design idea is just to separate them altogether.~\citet{HOSY17} shows that this category is cartesian closed, and~\citet{10.1145/3290349} shows how to interpret a rich type theory by exploiting an elegant integration with axiomatic domain theory.  This approach addresses the full expressiveness of statistical probabilistic languages.  However, since we are forced out of the realm of standard measure theory, we are perforce forced to ``reinvent the wheel'' for probabilistic reasoning.   The full integration with classical reasoning techniques of programming languages remains to be resolved in future work.    

Our suggested approach for \statprob-languages inherits powerful coinductive reasoning from open bisimulation, approximation techniques from the underlying domain theory and its intimately related metric space, and stays within the confines of classical measure theory.  However, as it stands, this paper is merely a promise.  Its full realization via a thorough study of a full language with continuous probability distributions in this setting is left to future work.    

\paragraph{Samson's impact on this article. } 
This article falls squarely into the areas pioneered by Samson, starting with the origins of open bisimulation from his study of the lazy lambda calculus and game semantics.  We conclude this introduction by tracing the (personal) impact of Samson's work even in the matter of probabilistic methods.   My interest in probabilistic computation is inspired by my thesis advisor and friend, Prakash Panangaden, an energetic and enthusiastic advocate~\citep{10.5555/1717290} of probabilistic methods in semantics.  In analogy to ~\citet{ABRAMSKY1991161}'s analysis of strong bisimulation,~\citet{DBLP:journals/iandc/DesharnaisGJP03} explores a domain equation for strong probabilistic bisimulation.~\citet{DBLP:journals/iandc/DesharnaisGJP03}'s study of simulation and this article's analysis of the simulation relation in this article draw heavily on the intuitions inspired by~\citet{ABRAMSKY19911}.

%% file: language.tex
\section{$\Lambda$: probabilistic lazy lambda calculus}

We study a lazy lambda calculus with countable formal sums.  Despite the title of this section, at this stage, we do not associate formal probabilities with the terms; instead, we view the formal sums of terms merely as weighted terms.   

Our technical definitions follow the general style of~\citet{10.1145/3350618}, adapted to our setting of open bisimulation.   

\subsection{Syntax}

\begin{rdisplaytab}[\renewcommand{\displayratio}{.23}]{ $\Lambda$.  Syntax.}{}{}\label{syntax}
  \rcategoryset{\aprob,\bprob}{Probabilities in $\SBR{0 \ldots 1}$}{}
  \\[1ex]
  \rcategoryalt{\atrmm,\btrmm,\ctrmm}{Distributions of terms}{}
   \rentry{\CBR{\pterm{\aprob}{\atrm}}}{\aprob}{\fn{\atrm}}
   \rentry{\union{\atrmm}{\btrmm}}{\probfull{\atrmm}+ \probfull{\btrmm}}{\fn{\atrmm} \cup \fn{\btrmm}}
\\[1ex]
  \rcategoryset{\avar,\bvar,\aprc,\asym,\bsym,\csym}{Variable Names}{}
  \\[1ex]
  \rcategoryalt{\atrm,\btrm,\ctrm}{Terms}{}
  \rentry{\avar}{}{\CBR{\avar}}
  \rentry{\abs{\avar}{\atrmm}}{}{\fn{\atrmm} \setminus \CBR{\avar}}
  \rentry{\atrmm\A\btrmm}{}{\fn{\atrmm} \cup \fn{\btrmm}}
\\[1ex]
   \rcategoryalt{\atrmms,\btrmms}{Vector of finite distributions of terms }{}
\\[1ex]
    \rcategoryalt{\atrmset,\btrmset}{Sets of terms}{}
\end{rdisplaytab}
$|\atrmms|$ is the length of the vector.   We use $\atrmms\PBR{i}$ for the $i$'th term of $\atrmms$ when $1 \leq i \leq |\atrmms|$.   

$\fn{\atrmm}$ is the set of free names of $\atrmm$.  As usual, we consider terms upto renaming of bound variables.

$\prob{\atrmm}{\btrmm} = \aprob$ if $\pterm{\aprob}{\btrmm} \in \atrmm$, and $0$ otherwise.  $\probfull{\atrmm}$ is the full weight of $\atrmm$.  We will only consider $\atrmm$ such that $\probfull{\atrmm} \leq 1$.   

Let $\aprob \times \btrmm = \CBR{   \pterm{\aprob \times \bprob_i}{\atrm_i} \mid \pterm{\bprob_i}{\atrm} \in \btrmm}$
If $\sum_i \aprob_i \leq 1$, then the {\em sub-convex combination} $\CBR{ \aprob_i  \times \atrmm_i}$ is a valid multiset of terms.  An important special case is written $\union{\atrmm}{\btrmm}$, that stands for for the weighted term such that $\prob{\union{\atrmm}{\btrmm}}{\ctrmm}= \prob{\atrmm}{\ctrmm} + \prob{\btrmm}{\ctrmm}$, if it exists.  Traditional probabilistic choice,  written $\union{\pterm{0.5}{\atrm}}{\pterm{0.5}{\btrm}}$, thus stands for $\{ \pterm{0.5}{\atrm}, \pterm{0.5}{\btrm} \}$.  

The above grammar only generates finite distributions.  The forthcoming reduction semantics can generate countable distributions; we will also use $\atrmm,\btrmm$ for the countable distributions generated by the reduction semantics.  

\begin{definition}[Ordering distributions]
$\atrmm \leq \btrmm\ \defeq\  (\exists \ctrmm)  \ [ \btrmm = \union{\atrmm}{\ctrmm}]$.
\end{definition}
Thus, if $\atrmm \leq \btrmm$, then $(\forall \pterm{\aprob}{\atrm} \in \atrmm) \ (\exists \pterm{\bprob}{\atrm} \in \btrmm) \ \aprob \leq \bprob$.    Under this order, distributions with countable support form a bounded complete, continuous dcpo.  Bounds do not exist only when total weight adds up to more than $1$.  The least upper bounds of bounded sets, written $\bigsqcup \atrmm_i$, given by $\CBR{\pterm{\sup\ \aprob_i}{\atrm} \mid \pterm{\aprob_i}{\atrm} \in \atrmm_i}$ and a countable basis given by distributions with finite support and rational weights, with the ``way-below'' relation given by:
\[\atrmm \ll \btrmm\  \defeq\  |\atrmm| \text{ finite },  (\forall \pterm{\aprob}{\atrm} \in \atrmm) \ (\exists \pterm{\bprob}{\atrm} \in \btrmm) \ \aprob < \bprob 
\]
   

We will often avoid some explicit coercions between terms and their associated point measures, \eg, we will often write $\CBR{\pterm{1}{\atrm}}$ simply as $\atrm$.

\subsection{Reduction semantics}
We consider a lazy reduction strategy to weak head normal forms.  For pure lambda terms, this reduction strategy corresponds to Levy-Longo trees.   

The \emph{evaluation
  contexts} are defined as follows.
\begin{displaymath}
  \actx,\bctx\BNFDEF
  \ahole \BNFSEP \actx\A\atrmm
\end{displaymath}

The small step evaluation relation $\atrmm \eval{} \btrmm$ below tracks the full ensemble of possible results.   The evaluation  relation ensures that application is left linear.  Since the evaluation is call-by-name, any choices in the argument are resolved separately and independently at the point of use of the argument.   We write $\subst{\atrmm}{\avar}{\btrmm}$ for the usual capture free substitution lifted to sets.   
\begin{scope}
  \begin{displaytab}[\renewcommand{\displayratio}{.33}]{Evaluation
      \;\;$(\atrmm \teval{} \btrmm)$.
  $\evals{}$: transitive closure of $\eval{}$.
    }
    \rinfer{\abs{\avar}{\atrmm} \A \btrmm
      \teval{}
      \subst{\atrmm}{\avar}{\btrmm}
    }   
    \\[.5ex]
    \linfer{}{\atrmm \teval{} \atrmm'}
       {\actx[\atrmm] \teval{} \actx[\atrmm']}
   \\[.5ex]
    \rinfer{\atrmm \A \btrmm
      \eval{}
      \biguplus_{\pterm{\aprob}{\atrm} \in \atrmm} \CBR{\pterm{\aprob}{\atrm \A \btrmm}}
    }   
    \\[.5ex]
    \linfer{}{\atrm \eval{} \btrmm}
     {\union{\pterm{\aprob}{\atrm}}{\atrmm}
      \eval{}
      \union{\PBR{\aprob \times \btrmm}}{\atrmm}
      }
  \end{displaytab}
\end{scope}

We summarize some properties of $\eval{}$.  
\begin{itemize}
\item $\eval{}$ is one-step Church Rosser since there are no critical pairs.  
\item If $\atrmm \leq \btrmm$ and $\atrmm \eval{} \atrmm'$, then there exists $\btrmm'$ such that $\btrmm \eval{} \btrmm'$ and $\atrmm' \leq \btrmm'$
\item Let $\atrmm_i$ be a directed set.  Let $\btrmm_i$ be a directed set such that $\atrmm_i \eval{}{} \btrmm_i$.   Then, $\bigsqcup \atrmm_i \eval{} \bigsqcup \btrmm_i$.
\item $\eval{}$ is linear wrt $\uplus$.  If $\atrmm = \biguplus_i \atrmm_i$ and $ (\forall i)  \atrmm_i  \evals{} \btrmm_i$, then $\atrmm\ \evals{} \biguplus_i \btrmm_i$.  
\end{itemize}
$\evals{}$ inherits the linearity and monotonicity  properties from $\eval{}$. 

\begin{definition}[Weak head normal forms]
The weak head normal forms are of the form $\pterm{\aprob}{\abs{\avar}{\atrmm}}$ and $\pterm{\aprob}{\avar \A \atrmms}$.
\end{definition}
We use $\vals{\atrmm}$ for the sub distribution of $\atrmm$ with support only on the (weak head) normal forms in $\atrmm$.   
\[ \vals{\atrmm} =  \CBR{ \pterm{\aprob}{\atrm} \in \atrmm \mid \atrm \mbox{ in } \whnf{}} \]
$\vals{}$ is monotone: if $\atrmm \eval{} \btrmm$, then $\vals{\atrmm} \leq\vals{\btrmm}$.   The result of big-step evaluation is a distribution on normal forms.     
\begin{definition}[Big step evaluation] 
$\whnf{\atrmm} = \bigsqcup \CBR{ \vals{\ctrmm} \mid \atrmm \evals{} \ctrmm}$.  
\end{definition}

This evaluation might be necessarily infinite, as shown by the following example, reproduced from the introduction.  
\begin{example}\label{infinite}
Let $\ycomb$ be a fixpoint combinator such that $\ycomb \A \atrm\ \evals{} \atrm \A \PBR{\ycomb \A \atrm}$, $\Id$ be $\abs{\avar}{\avar}$, and let $\atrm = \abs{\avar}{ \CBR{\pterm{0.5}{\Id}, \pterm{0.5}{\avar}}}$.  Then, $\cvg{\ycomb \A \atrm}{\Id}$ but for any finite evolution $\ycomb \A \atrm\ \evals{} \btrmm$, $\vals{\btrmm}$ necessarily assigns a probability strictly less than $1$ to $\Id$.
\end{example}
Following ~\citet{10.1145/3350618}, we generally use big-step evaluation for definitions and small-step evaluation for proofs.  

We summarize the linearity and monotonicity properties of big-step evaluation.
\begin{itemize}
\item If $\atrmm \leq \btrmm$,  then $\whnf{\atrmm} \leq \whnf{\btrmm}$.  
\item Let $\atrmm_i$ be a directed set.  Then, $\whnf{\bigsqcup \atrmm_i} = \bigsqcup \whnf{\atrmm_i}$.
\item If $\atrmm = \biguplus_i \atrmm_i$, then $\whnf{\atrmm} = \biguplus_i \whnf{\atrmm_i}$.
\end{itemize}

\section{Open Simulation}

\subsection{LTS basis for open simulation}

We present an LTS on terms as a prelude to defining a notion of simulation on terms, building on~\citet{LASSEN1999346}'s presentation of ~\citet{SANGIORGI1994120}'s treatment.  In particular, we follow the ``open bisimulation'' presentation of Lassen's later work~\citep{DBLP:conf/lics/Lassen05,DBLP:conf/lics/LassenL08,DBLP:conf/birthday/StovringL09} in a form that validates $\eta$-expansion~\citep{DBLP:journals/taosd/JagadeesanPR09}.

\begin{displaytab}[\renewcommand{\displayratio}{.23}]{LTS
    Labels}
  \hspace{1em} \= \hspace{.10\linewidth-1em} \= \hspace{1em} \= \kill
  \categoryone{\albltau}{}
  {\tau_{} \BNFSEP\albl}{All Labels}
  \\[.5ex]
  \categoryone{\tau}{Silent Label}{}
  \\[.5ex]
  \categoryalt{\albl}{Visible Labels}
  \entry{\ltsfuncall{\asym}{\PBR{i,n}}}{Term uses arg $i$ out of $n$ of context function $\asym$}
  \entry{\ltsapp{\asym}}{Context calls term with argument $\asym$ ($\dn{}=\set{\asym}$)}
  \entry{\ltssignal}{Convergence}
\end{displaytab}

The choice of labels is determined by the possibilities available to the context to interact with the term.   In~\citet{Gordon95}'s terminology, the visible label $\ltsfuncall{}{}$ has an \emph{active} component (representing actions initiated by the term), whereas $\ltsapp{\asym}$ is completely \emph{passive} (representing actions initiated by the environment).  $\ltssignal$ is also passive.  

The silent label $\tau$ stands for internal computation of the term.

\begin{scope}
\begin{displaytab}{LTS on terms \;\; $\ltsstrongtransition{\atrm}{\albltau}{\btrmm}$}
\linfer{SILENT}{
      \atrmm \eval{} \btrmm
    }{
      \ltsstrongtransition{
        \atrmm}{\tau}{
        \btrmm}
    }
  \\[1ex]
  \linfer{RETURN}{}{
      \ltsstrongtransition{\abs{\avar}{\atrm}}
        {\ltsret{\asym}
        }{
        (\abs{\avar}{\atrm}) \A \asym}}
  \\[1ex]
  \linfer{CONVERGE}{}{
      \ltsstrongtransition{\abs{\avar}{\atrm}}
        {\ltssignal}
        {
        \Id}}
  \\[1ex]
  \linfer{CALL}{0 \leq i \leq |\atrmms|}
     {  \ltsstrongtransition{\asym \A \atrmms }
                                   {\ltsfuncall{\asym}{\PBR{i,|\atrmms|}}}
                                  {\atrmms\PBR{i}}
}
\\
\\
\linfer{LIFT-$\albl$}
       {\ltsstrongtransition{\atrm}{\albl}{\btrm}}
       {\ltsstrongtransition{\pterm{\aprob}{\atrm}}{\albl}{\pterm{\aprob}{\btrm}}}
 \\[1ex]
  \end{displaytab}
\end{scope}

The LTS affords priority to the  internal reductions of the  term, since the visible transitions  are only applicable to weak-head normal forms.  

$\ltsapp{\asym}$ performs applicative tests.  Rather than providing a term as an argument for the applicative test, this rule provides a
(possibly fresh) symbolic argument $\asym$.   $\ltssignal$ facilitates the measurement of the weight of convergence to an abstraction.  

$\ltsfuncall{\asym}{\PBR{i,n}}$ records the call to an external and unknown parameter $\asym$.  It chooses an argument, the one at the $i$'th position, to inspect further.  The addition of parameter $n$ provides a way to distinguish argument vectors of different lengths.   We adopt the convention that $\atrmms\PBR{0}$  is $\Id$.  Mirroring $\ltssignal$, the label $\ltsfuncall{\asym}{\PBR{0,n}}$ is a way to measure the weight of convergence to an head application of $\asym$ to $n$ arguments.  

Notice that the \RN{lts} is being defined on terms rather than formal sums of terms.  This is because of a vexing difference between the usual weak-head normal forms $\abs{\avar}{\atrmm}$ and the ones introduced by open application $\asym \A \atrmms$, that prevents us 
from following the correct approach of directly building a labeled Markov chain by lifting these LTS transitions.

\begin{example}[Linearity and Weak head normal forms]\label{lis}
Application is linear in the head variable.  This is reflected in a candidate \RN{LTS} extension to formal sums of abstractions as follows:
\[\frac{\atrm_i = \abs{\avar}{\atrmm_i}, \ \ltsstrongtransition{\atrm_i}{\albl} {\btrmm_i}}{\ltsstrongtransition{\union{\atrmm}{\btrmm}}{\albl}{\uplus_i \btrmm_i}}
\]
This rule permits us to show that abstraction is linear wrt formal sums.  \ie.
\[ \abs{\avar}{\union{\atrmms}{\btrmms}} \bisimreln\ \union{\abs{\avar}{\atrmms}}{\abs{\avar}{\btrmms}} \]
where the matching of transitions of the left by the right makes essential use of the lifting of transitions on terms to transitions on formal sums of terms.

Alas, this lifting is not sound for formal sums of the form $\union{\asym \A \atrmms}{\asym \A \btrmms}$.  To see this, consider $\union{(\avar \A \tru\ \fls)}{(\avar \A \fls\ \tru)}$ and $\union{(\avar \A \fls\ \fls)}{(\avar \A \tru\ \tru)}$, where $\tru, \fls$ are the Church booleans.    Setting $\avar$ to be the Church terms for exclusive-OR differentiates them.  However, both formal sums have transitions with labels $\ltsfuncall{\avar}{\PBR{i,|2|}}$ to $\union{\tru}{\fls}$ for $i \in \{1,2\}$, if we  permitted unrestricted linearity of transition system.  
\end{example}

\subsection{Probabilistic lifting}
In order to establish the proper foundations for our forthcoming definition of simulation, we take a short detour into lifting of relations to probability distributions.   We try to make the material in this subsection self-contained.  

Let $\areln$ be a binary relation on a countable $\cset$.  We use $\adist,\bdist$ for discrete sub probability distributions, and $|\adist|, |\bdist|$ for their carrier sets.   We write $\sum_i \aprob_i \times \aElt_i$ for the sub-convex combination of point measures at $\aElt_i$.   Lift binary relations on $\cset$ are to probability distributions as follows.    
\begin{definition}[Probabilistic lift]
Let $\adist = \sum_i \aprob_i \times \aElt_i,\bdist = \sum_j \bprob_j \times \bElt_j$.  Then:
$\adist\   \close{\areln} \ \bdist$ if there is a matching $\fm_{i,j}$ such that:
\begin{align*}
\fm_{i,j} \neq 0 &\Rightarrow \aElt_i\ \areln\ \bElt_j \\
\sum_i \fm_{i,j} &= \aprob_i  \\
\sum_j \fm_{i,j} &\leq \bprob_j 
\end{align*}
\end{definition}
In contrast to the definitions of~\citet{LS,10.1109/LICS.2007.15}, the inequality in the last line above permits us to address subprobabilities.

The following lemma shows that $\close{\areln}$ can be understood as a Kleisli construction on the  probabilistic power domain of~\citet{10.5555/77350.77370}.  This``splitting lemma'' for (countable) probability distributions is a small generalization from the finite in Chapter 4 of~\citet{10.5555/77350.77370}.  
\begin{lemma}\label{split}
$\adist\   \close{\areln} \ \bdist$ iff for all $\csub \subseteq\ |\adist|$ such that
$(\csub\ \SEMI\ \areln) \bigcap\ |\adist| = \csub$, it is the case that:
\[ \sum_{\aElt_k \in\ \csub} \aprob_k  \leq \sum_{\bElt_j \in |\bdist|} \bprob_j \mid  \bElt_j \in\ \csub \SEMI\ \areln \]
\end{lemma}
\begin{proof}
We follow  Chapter 4 of~\citet{10.5555/77350.77370} essentially verbatim, using the countable version of MaxFlow-MinCut  from~\citet{AHARONI20111}; in any (possibly infinite) network there exists an orthogonal pair of a flow and a cut\footnote{A flow and a cut are orthogonal if the flow exhausts the capacity of the edges that go forward in the cut, even as it assigns $0$ flow to the edges that flow back in the cut.}. We use this result for very simple countably infinite directed bipartite graph with {\em no} infinite paths.

The vertices  of the graph are given by $\CBR{\src,\sink} \cup\ |\adist| \cup |\bdist|$.  $\src$ is connected to $\aElt_i$ with capacity $\aprob_i$. $\bElt_i$ is connected to $\sink$ with capacity $\bprob_i$.  $\aElt_i$ is connected to $\bElt_j$ with a very high capacity, say $100$.   

Consider \PBR{\aFlow,\aCut},  the orthogonal pair of a  flow and a cut in this graph.  

If all the edges from the $\src$ are in \aCut, \aFlow\ validates the criterion, yielding the result for the forward direction.  

If not.  Since the capacity of edges from $|\adist|$ into $|\bdist|$ is large, they cannot be exhausted.   So, \aCut\ consists only of edges from the $\src$ or into the $\sink$.  The set of vertices connected to $\src$ in \aCut\ provides evidence of violation of the criterion. 
\end{proof}
In the forthcoming development, we use the matching definition in the precongruence proof and the characterization of lemma~\ref{split} almost everywhere else.  

Since $\areln\ \SEMI\ \bigcap_i\ \areln_i = \bigcap_i \areln\ \SEMI\ \areln_i$, we deduce:
\begin{corollary}\label{omegaint}
\[ \close{\bigcap_i\ \areln_i} = \bigcap_i \ \close{\areln} \]
\end{corollary}

Let $\sum_i \aprob'_i \times \aElt_i \leq\ \sum_i \aprob_i \times \aElt_i$  if $\aprob'_i < \aprob_i$ for all $i$.  It is immediate that:
\begin{corollary}\label{lldist}
\[ \adist\   \close{\areln} \ \bdist\ \Longleftrightarrow\ (\forall \adist' \ll \adist) \adist'\  \close{\areln} \ \bdist   \]
\end{corollary}

\subsection{Simulation via coinduction}

$\lnf \subseteq \Lambda$ is the subset of weak-head normal forms. $\lnfa$ restricts to formal sums of only abstractions.  $\lnfo$ restricts to  point formal sums of open applications.
\begin{definition}[$\lnf$]
\begin{align*}
\atrmm \in \lnf &\defeq\ [ \atrmm \evals{} \btrm  \Rightarrow\ \atrmm = \btrm ] \\
\atrmm \in \lnfa &\defeq\ [ \forall \pterm{\aprob}{\atrm} \in \atrmm, \atrm \ \mbox{is an abstraction}] \\
\atrmm \in \lnfo &\defeq\ [ (\exists \aprob, \atrmms, \avar) \ \atrmm = \pterm{\aprob}{\avar \A\ \atrmms} ] \\
\relspace &\defeq (\lnfa \times \lnfa) \bigcup\ (\lnfo \times \lnfo)
\end{align*}
\end{definition}
The very type of $\relspace$ addressses the linearity issue of example~\ref{lis}. 

\begin{definition}[Weak max transitions: $\weakmaxtransition{\albltau}{}: \subseteq \lnf \times \lnf$]
\item $\atrmm \weakmaxtransition{\tau }{} \whnf{\atrmm}$.
\item $\atrmm \weakmaxtransition{\albl}{} \ctrmm$, if  $ \atrmm \strongtransition{\albl}{} \btrmm \weakmaxtransition{\tau }{}\ctrmm $
\end{definition}  
The unique target of $\weakmaxtransition{\albltau}{}$ exemplifies the probabilistic determinacy of our language. 

We consider binary relations of terms that have a kernel in $\relspace$.
\begin{definition}
$\rels$ is the set of relations $\areln$ over $\Lambda$ that are induced by a kernel in $\relspace$ as follows:
\[ \areln = \CBR{\PBR{\atrmm,\btrmm} \mid\ \PBR{\whnf{\atrmm},\whnf{\btrmm}} \in\ \close{\areln \cap\ \relspace}} \]
\end{definition}
$\rels$ is a complete lattice under subset with maximum element induced by $\relspace$ and arbitrary least upper bounds induced by the union of their  kernels.   Thus, our definition of simulation fits in the framework described by~\citet{10.1145/2933575.2934564} for simulation relations and upto-reasoning.
\begin{definition}[Simulation Functional] \hfill \\
Define a monotone operator $\simfn{}$ on $\rels$ as follows.  Let $\areln \in \rels$. Let $\PBR{\atrmm,\btrmm} \in \relspace$.  

$\atrmm\ \simfn{\areln}\ \btrmm$ if whenever $\atrmm \weakmaxtransition{\albltau}{} \atrmm'$ and $ \btrmm \weakmaxtransition{\albltau}{} \btrmm'$,  it is the case that  $ \atrmm'\ \close{\areln} \ \btrmm'$.
\end{definition}
In the above definition, we only define $\simfn{\areln} \subseteq \relspace$, since the extension to $\rels$ is unique.  
There is a maximum simulation, that we write as $\simreln$.  Let $\bisimreln = \simreln \cap \simreln^{-1}$.  $\simreln$ is extended to $\Lambda \times \Lambda$ by extension by closing and reduction.

We list out some useful properties of $\simfn{}$.   As a consequence of the inequality in the definition of $\close{\areln}$, we deduce that:
\[ \leq \SEMI \simfn{\areln} \SEMI  \leq = \simfn{\areln} \]
and
\[\close{\simfn{\areln}} = \simfn{\areln} \]
$\simfn{\areln}$ is transitive if $\areln$ is.  

From lemma~\ref{lldist}, $\simfn{\areln}$ is admissible, in the sense of~\citet{10.5555/218742.218744}, for $\leq$-chains, \ie:
if $\CBR{\atrmm_i}$ is a $\leq$-chain,  and $(\forall i)  \ \atrmm_i\ \simfn{\areln}\  \btrmm$, then $\bigsqcup \CBR{\atrmm_i} \simfn{\areln} \btrmm$.   This provides a finitary operational perspective that is crucial to the forthcoming precongruence proof.  A motivating example is to show that $\Id \simreln\  \ycomb \A \atrm$ from example~\ref{infinite}; notably, none of the finite unwindings of $\ycomb \A \atrm$ suffice to reach the limit $\Id$.   Since every finite approximant to $\atrmm$ can be achieved in a finite computation, admissibility ensures that in order for $\atrmm$ to be simulated by $\btrmm$, it suffices for the finite computations from $\atrmm$ to be  simulated by (finite computations of) $\btrmm$.   

Finally, as a consequence of probabilistic determinacy and lemma~\ref{omegaint}, we deduce that the closure ordinal of $\simreln$ is $\omega$; \ie, let
\[
\begin{array}{lll}
\simreln_0 &=& \Lambda \times \Lambda \\
\simreln_{k+1} &=& \simfn{\simreln_k}
\end{array}
\]
Then, $\simreln = \bigcap_k \simreln_k$.

These properties translate to  $\simreln$ as follows.  
\begin{itemize}
\item $\simreln$ is a preorder with least element $\least$, where $\least$ is any term that does not converge to a weak head normal form, \eg, $\PBR{\abs{\avar}{\avar \A \avar}} \A  \PBR{\abs{\avar}{\avar \A \avar}}$. 
\item $\PBR{} \bisimreln \least$.  In probabilistic programming languages, divergence causes loss of probability.
\item $\leq \SEMI\ \simreln\ \SEMI\ \leq = \simreln$

\item $\atrmm \bisimreln\ \whnf{\atrmm}$;~\citet{10.1145/3350618} call this principle ``simulation up-to lifting''. 
\item Simulation is closed under under sub-convex combinations, \ie,  if $\atrmm_i \simreln \btrmm_i$ forall $i$, then  $\biguplus_i \aprob_i \times \atrmm_i  \simreln \biguplus_i \aprob_i \times \btrmm_i $; thus, $\simreln$ is a precongruence for choice.
\end{itemize}

\subsection{Precongruence proofs}

We will work with substitutions generated by the following grammar. 
\begin{rdisplaytab} [\renewcommand{\displayratio}{.23}]{Substitutions}{}{}
  \rcategoryalt{\asub,\asub'}{Substitutions}{}
      \rentry{\subst{}{\avar}{\atrmm}}{}{}
      \rentry{\asub \A \asub'}{}{}
\end{rdisplaytab}
We will use $\dom{\asub}$ for the domain of a substititution $\asub$.  We write $\rundef{\asub\PBR{\avar}}$ if $\avar \not \in \dom{\asub}$, and $\defined{\asub\PBR{\avar}}$ otherwise.    We will only use substitutions that satisfy the following restriction:
\[ (\forall  \avar \in \dom{\asub}) \ \avar\not \in \bigcup_{\bvar \in \dom{\asub}} \fn{\asub(\bvar)} \]

The subsititutive version of simulation is defined as follows.  
\begin{definition}
$\asub \simreln \bsub$ if $\dom{\asub} = \dom{\bsub}$ and for all $\avar \in \dom{\asub} \ [ \asub\PBR{\avar} \simreln\ \bsub\PBR{\avar} ] $. 

$\substsimreln$ is defined as the smallest relation that satisfies 
for all $\atrm, \atrm',\asub,\bsub$, 
\[ \atrmm \simreln \atrmm', \asub \simreln \bsub \Longrightarrow\ \atrmm \A \asub   \substsimreln \atrmm'  \A \bsub \]
\end{definition}
In particular,  $\simreln \subseteq \substsimreln$.  

In the appendix~\ref{proof}, we prove that $\substsimreln$ is a postfixed point of $\simfn{}$.    
\begin{lemma}\label{substsim}
$\substsimreln  \subseteq \simfn{\substsimreln}$.
\end{lemma}
\begin{theorem}\label{simthm}[Simulation is a precongruence] \hfill \\
$\simreln$ is a precongruence for all program combinators.
\end{theorem}
\begin{proof}
The proof follows from definitions for abstraction and $\uplus$.  For application, we are given $\atrmm_i, \btrmm_i$ for $i=1,2$ with $\atrmm_1\ \simreln\ \atrmm_2$ and $\btrmm_1\ \simreln\ \bsub_2$.   Use above lemma~\ref{substsim} with the term $\avar \A \bvar$ and substitutions $\asub_i$, for $i=1,2$ such that $\asub_i\PBR{\avar} = \atrmm_i, \asub_i\PBR{\bvar} = \btrmm_i$.
\end{proof}

%% file: domain.tex
\section{Building a Polish space}

We adapt the definition of ``compact'' trees of~\citet{DBLP:conf/icalp/Ong92} to a setting with weighted terms, by incorporating distributions with finite support .  

\begin{rdisplaytab}[\renewcommand{\displayratio}{.23}]{ $\lfin$}{}{}
  \rcategoryset{\aprob}{Rational Probabilities in $\SBR{0 \ldots 1}$}{}
  \\[1ex]
   \rcategoryalt{\ctrmm,\dtrmm}{Finite Distributions of terms}{}
   \rentry{\CBR{\pterm{\aprob}{\ctrm}}}{}{}
   \rentry{\union{\ctrmm}{\dtrmm}}{}{}
\\[1ex]
 \rcategoryalt{\atrm}{Terms}{}
  \rentry{\Omega}{}{}
  \rentry{\abs{\avar}{\atrm}}{}{}
  \rentry{\bvar \A \ctrmms}{}{}
\\[1ex]
   \rcategoryalt{\ctrmms}{Finite vector of finite distributions of terms}{}
\end{rdisplaytab}
When restricted to Longo trees, the definition coincides with~\citet{DBLP:conf/icalp/Ong92}, also Chapter~2 of~\citet{Ong88}.  In this section, we will restrict the use of $\ctrmm, \dtrmm, \ldots $ for elements of $\lfin$, whereas $\atrmm,\btrmm, \ldots$ will be used for generaal terms of $\Lambda$.  

\begin{definition}[Approximants]\label{approximants}
Let $\atrmm \in \Lambda$.  For each $n \geq 0$, define  $\nth{\atrmm}{n} \subseteq \lfin$, inductively as follows.
\[
\begin{array}{llll}
\least  & \in &\hspace*{.1in} & \nth{\atrmm}{0} \\
\nth{\atrmm}{k}& \subseteq&&\nth{\atrmm}{k+1} \\
\ctrmm &\in&&\nth{\atrmm}{k}, \ \text{ if }  \ctrmm \in  \close{\nth{\atrmm}{k}} \\
\ctrmm &\in && \nth{\atrmm}{k}, \ \text{ if }   \ \btrmm\ \ll\ \whnf{\atrmm}, \ \ctrmm\ \in \nth{\btrmm}{k} \\
\ctrmms & \in && \nth{\atrmms}{k}, \ \mbox{ if } |\ctrmms| = | \atrmms|, \ (\forall 1 \leq i \leq |\atrmms|) \  \ctrmms\PBR{i} \in \atrmms\PBR{i} \\
\pterm{\aprob}{\abs{\avar}{\ctrmm}} &\in&& \nth{\pterm{\bprob}{\abs{\avar}{\atrmm}}}{k+1}, \ \ \mbox{ if } \aprob < \bprob, \aprob \times \ctrmm \in \nth{\bprob \times \atrmm}{k} \\
\pterm{\aprob}{\avar \A \ctrmms} &\in&& \nth{\pterm{\bprob}{\avar \A \atrmms}}{k+1}, \ \ \mbox{ if } \aprob < \bprob, |\atrmms| = |\ctrmms|, \ (\forall 1 \leq i \leq |\atrmms|) \  \aprob \times \ctrmms\PBR{i} \in \nth{\bprob \times \atrmms\PBR{i}}{k}
\end{array}
\]
\end{definition}

\begin{definition}
$\ctrmm \lll \atrmm$ if $(\exists k) \ \ctrmm \in \nth{\atrmm}{k}$
\end{definition}

We identify the key properties of $\lll$ following proposition 2.3 of \citet{LAWSON1998247}.
\begin{lemma}\label{lawson1}
Let $\ctrmm,\dtrmm \in \lfin$.  
\begin{align}
\ctrmm\ \simreln\  \atrmm\ &\ \Longleftarrow\ \  \ctrmm\ \lll\  \atrmm  \label{eqn:1}\\
\ctrmm\ \lll\ \btrmm\ &\ \Longleftarrow\ \  \ctrmm\ \simreln\ \dtrmm\ \lll\  \atrmm\ \simreln\  \btrmm \label{eqn:2} \\
\ctrmm\ \lll\ \atrmm\ &\ \Longrightarrow\ \  (\exists \dtrmm) \ctrmm \lll \dtrmm \lll \atrmm \label{eqn:3}
\end{align}
\end{lemma}
\begin{proof}
All proofs proceed by routine induction on $k$ such that $\ctrmm \in \nth{\atrmm}{k}$.   
\end{proof}

$\lll$ determines the simulation order.  
\begin{lemma}\label{simvialll}
\[ \atrmm\ \simreln\ \btrmm\ \Longleftrightarrow\  (\forall \ctrmm) \ [ \ctrmm\ \lll\ \atrmm\ \Longrightarrow\ \ctrmm\ \simreln\ \btrmm] 
\]
\end{lemma}
\begin{proof}
The forward direction follows from equation~\ref{eqn:1} of lemma~\ref{lawson1} and transitivity of $\simreln$.  

For the converse, we prove by induction on $k$ that:
\[
(\forall k) \ [ \atrmm\ \simreln_k \ \btrmm] \Longleftarrow\ \ (\forall \ctrmm\ \in\ \nth{\atrmm}{k}) \  [\ctrmm\ \simreln\ \btrmm]
\]
The base case is immediate.  

Consider the inductive case at $k+1$.  
Since $\close{\simreln} = \simreln$, and $\close{\lll} = \lll$, it suffices to prove for $\atrmm,\btrmm, \ctrmm$ such that $\PBR{\PBR{\ctrmm,\atrmm},\PBR{\ctrmm,\btrmm}} \subseteq\ \relspace$.  Simplifying a bit further, using linearity of abstraction wrt formal sums, we deduce that it suffices to consider just the following two cases.

\begin{itemize}
\item $\atrm = \pterm{\aprob}{\avar \A \ \atrmms}, \btrm = \pterm{\bprob}{\avar \A \ \btrmms}$, where $|\atrmms| = |\btrmms| $.  We are aiming to prove $\atrm\  \simreln_{k+1} \ \btrm$.  

Consider $\ctrm = \pterm{\aprob'}{\avar \A \ \PLURAL{\Omega}}$, for any $\aprob' < \aprob$.  $\ctrm \in \nth{\atrmm}{1} \subseteq \nth{\btrmm}{1}$.  So, we deduce that $\aprob \leq \bprob$.

Let $1 \leq i \leq |\atrmms|$.   We need to show that $\atrmms\PBR{i} \  \sim_k\ \btrmms\PBR{i}$.  By induction hypothesis, it suffices to show that $\nth{\atrmms\PBR{i}}{k} \subseteq\ \nth{\btrmms\PBR{i}}{k}$.  Let $\dtrmm \in \nth{\atrmms\PBR{i}}{k}$.  Then, $\ctrm = \pterm{\aprob'}{\avar \A \ctrmms} \in \nth{\btrmms\PBR{i}}{k+1}$, where $\aprob' < \aprob$ and $\ctrmms\PBR{j} = \dtrmm, i=j$ and $\Omega$ otherwise.  By induction hypothesis, $\ctrm \in\  \nth{\btrmms}{k+1}$.  Thus, we deduce that $\dtrmm \in\ \nth{\btrmms\PBR{i}}{k}$.  

\item  $\atrm = \pterm{\aprob}{\abs{\avar}{\atrmm}}, \btrm = \pterm{\bprob}{\abs{\avar}{\btrmm}}$.  We are aiming to prove $\atrm\  \simreln_{k+1} \ \btrm$.  

Consider $\ctrm = \pterm{\aprob'}{\abs{\avar}{\Omega}}$, for any $\aprob' < \aprob$.  $\ctrm \in \nth{\atrmm}{1} \subseteq \nth{\btrm}{1}$.  Thus, $\aprob \leq \bprob$.

We need to show that $\atrmm \sim_k\ \btrmm$.  By induction hypothesis, it suffices to show that $\nth{\atrmm}{k} \subseteq \nth{\btrmm}{k}$.  Let $\dtrmm \in\nth{\atrmm}{k}$.  Then $\pterm{\aprob'}{\abs{\avar}{\dtrmm}} \in \nth{\atrmm}{k+1}  \subseteq\ \nth{\btrmm}{k+1}$.  Thus, we deduce that $\dtrmm \in\nth{\btrmm}{k}$.
\end{itemize}

\end{proof}
\begin{corollary}\label{lllsim}
$\atrmm \simreln \btrmm \Longleftrightarrow\  (\forall \ctrmm) \ [\ctrmm \lll \ \atrmm \Longrightarrow\  \ctrmm \lll \ \btrmm] 
$
\end{corollary}
\begin{proof}
The forward direction comes from equation~\ref{eqn:2} .  The backward direction follows from  equation~\ref{eqn:1} and lemma~\ref{simvialll}.
\end{proof}

\begin{definition}[Lawson topology]
$\lquot = \Lambda/{\scriptstyle \bisimreln}$. 
Let
\[
\begin{array}{lll}
\dup{\ctrmm} &=& \CBR{ \atrmm \mid \ctrmm \lll \atrmm} \\
\up{\ctrmm} &=& \CBR{ \atrmm \mid \ctrmm \simreln \atrmm} \\
\end{array}
\]
The \ltop\ on $\lquot$ has subbasic open sets $\CBR{\dup{\ctrmm}, \lawson{\ctrmm}  \ \mid\ \ctrmm \in \lfin}$.
\end{definition}

$\dup{\ctrmm}$ is up-closed wrt $\simreln$, and $\down{\atrmm}$ is closed in the above topology.  
\begin{lemma}\hfill 
\begin{itemize}
\item $\dup{\ctrmm} \SEMI \simreln = \dup{\ctrmm}$. 
\item $\down{\atrmm} = \CBR{ \btrmm \mid \btrmm \simreln \atrmm} = \bigcap  \CBR{ (\lquot \setminus \dup{\ctrmm}) \mid \lfin \ni \ctrmm\not\lll \atrmm} $ is closed.  
\end{itemize}
\end{lemma}
\begin{proof}
The first item follows from equation~\ref{eqn:2}.

Let $\ctrmm\not\lll \atrmm$.   $\lquot \setminus \dup{\ctrmm}$  contains $\atrmm$.  Let $\btrmm \simreln \atrmm$.  If $\ctrmm \lll \btrmm$, then by equation~\ref{eqn:2}, $\ctrmm \lll \atrmm$, a contradiction.  Thus,  $\simreln \SEMI (\lquot \setminus \dup{\ctrmm}) = (\lquot \setminus \dup{\ctrmm}) $, and we deduce that $\down{\atrmm} \subseteq \bigcap  \CBR{ (\lquot \setminus \dup{\ctrmm}) \mid \lfin \ni \ctrmm\not\lll \atrmm} $.  By corollary~\ref{lllsim}, for any $\btrmm \not\simreln \atrmm$, there is a $\dtrmm \in \lfin$ such that $\dtrmm \lll \btrmm, \dtrmm \not\lll\atrmm$.   Thus, $\down{\atrmm} \supseteq \bigcap  \CBR{ (\lquot \setminus \dup{\ctrmm}) \mid \lfin \ni \ctrmm\not\lll \atrmm} $.  
\end{proof}
We are now able to mimic the statement and proof of proposition 4.3 in~\citet{LAWSON1998247}.  
\begin{lemma}
$\PBR{\lquot,\text{\ltop}}$ is Hausdorff, regular, and second countable.
\end{lemma}
\begin{proof}
$\lfin$  is countable, so second countability follows.

To prove Hausdorff, let $\atrmm \not\simreln \btrmm$.  By lemma~\ref{simvialll}, there exists $\ctrmm$ such that $\ctrmm \lll \atrmm$ and $\ctrmm \not\simreln\ \btrmm$.  The required (subbasic) open sets are $\dup{\ctrmm}, \lawson{\ctrmm}$.

For regularity, it suffices to show that each subbasic open set
containing a point contains a closed neighborhood of the point.
\begin{itemize}
\item Let $\atrmm \in \dup{\ctrmm}$.  Then, there exists $\ctrmm \lll  \dtrmm \lll \atrmm$.  The required closed set is $\up{\dtrmm}$.

\item Let $\atrmm \in \lawson{\ctrmm}$. The required closed set is $\down{\ctrmm}$.
\end{itemize}
\end{proof}

As a corollary, $\PBR{\lquot,\text{\ltop}}$ is metrizable.  
\begin{definition}
Let $\close{\lquot}$ be the completion of the metric space $\PBR{\lquot,\text{\ltop}}$.
\end{definition}

$\close{\lquot}$ is a bounded Polish space; hence, compact.  Thus,  we are finally ready to formally view the weighted terms $\atrmm,\btrmm, \ldots$ as (countable) sub-convex sums of unit masses.  Furthermore, it also serves as an appropriate universe to interpret a language with fully general continuous distributions.

%% file: end.tex
\section{Conclusions and Future work}
This article establishes the basic ingredients of a model for statistical probabilistic programming languages.  In contrast to the current research, we develop the foundations using a combination of traditional tools, namely open bisimulation and probabilistic simulation.  

These ideas remain but a promise, until used in a thorough investigation of the semantics of a programming language in this paradigm.  Not least because the ``competing'' methods of Quasi Borel Spaces~\citep{10.1145/3290349,HOSY17} provide precisely such a thorough investigation that includes evaluation strategies (call by value vs call by name), recursive types and continuous distributions.  We take hope from the extant research that demonstrates that open bisimulation is amenable to all these features and also provides excellent accounts of state and parametricity.

The coinductive and metric foundations of our approach provide the opportunity to explore principled mechanisms for approximate reasoning.   The opportunities for robust and approximate reasoning abounds in \statprob\ languages, driven by the desire to accommodate symbolic reasoning engines. For example:
\begin{itemize}
\item Approximating continuous distributions by other continuous distributions  or finite distributions to simplify symbolic reasoning

\item Approximating infinite computations by (large enough) finite unwindings
\end{itemize}
This motivates the desire to investigate coinductive principles for approximate reasoning, in the spirit of~\citet{DBLP:journals/tcs/DesharnaisGJP04}'s explorations into metric bisimulations for (concurrent) labelled Markov chains.  

\subsection*{Acknowledgements. }
I wish to thank Samson Abramsky for two decades of mentoring and friendship.  

This material is based upon work supported by the National Science Foundation under Grant
No. 1617175. Any opinions, findings, and conclusions or recommendations expressed in this material
are those of the author and do not necessarily reflect the views of the National Science Foundation.

%% file: appendix.tex
\section{Proof of Substitution lemma}\label{proof}

We use the following notation.

\begin{itemize}
\item $\PBR{\aprobs,\abs{\avars}{\btrmms}}$ for the weighted term $\CBR{ \aprobs\PBR{i} \times \abs{\avar}{\btrmms\PBR{i}}}$.  Thus, $\PBR{\aprobs,\abs{\avars}{\atrmms}}$  has a $\ltssignal$ transition to $\pterm{\sum_i \aprobs\PBR{i}}{\abs{\avar}{\avar}}$ and $\ltsret{\asym }$ transition to $\CBR{ \aprobs\PBR{i} \times \subst{\btrmms\PBR{i}}{\avar}{\asym}}$ that we write as $\PBR{\aprobs,\subst{\btrmms}{\avar}{\asym}}$.  
\item $\atrmms\ \simreln\ \btrmms$, if $|\atrmms| = |\btrmms|$, and $(\forall 1 \leq i \leq |\atrmms|) \ \atrmms\PBR{i} \ \simreln\ \ctrmms\PBR{i}$
\end{itemize}

\begin{lemma}\label{basecase}
Let $\atrmm_1 \A \asub_1\ \substsimreln\ \atrmm_2 \A \asub_2 $, and $\whnf{\atrmm_1 \asub_1} = \atrmm_1 \asub_1$.  Then, 
$\atrmm_1 \A \asub_1\  \simfn{\substsimreln} \  \atrmm_2 \A \asub_2$
\end{lemma}
\begin{proof}
Since $\close{\simfn{\substsimreln}} = \simfn{\substsimreln}$, it suffices to consider the following three cases for the shape of $\atrmm_1 \asub_1$.

\begin{enumerate}
\item  $\atrmm_1 = \PBR{\aprobs,\abs{\avars}{\btrmms}}$

By alpha renaming, we can assume that $\avar$ is chosen such that $\rundef{\asub_i(\avar)}$, so that $\subst{\asub_2}{\avar}{\asym}$ and $\subst{\asub_1}{\avar}{\asym}$ are valid substitutions. Thus, 

$\atrmm_1 \asub_1 = \PBR{\aprobs,\abs{\avars}{\btrmms \A \asub_1}{}}$.  
The two immediate transitions are:
\begin{itemize}
\item Labeled $\ltssignal$ to $\pterm{\sum_i \aprobs\PBR{i}}{\abs{\avar}{\avar}}$.
\item Labeled $\ltsret{\asym }$ to $\subst{\PBR{\bprobs,\btrmms \A \asub_1}}{\avar}{\asym}$
\end{itemize}
By $\atrmm_1 \simreln \atrmm_2$,  $\PBR{\bprobs,\abs{\avars}{\btrmms'}} \leq\ \whnf{\atrmm_2}$, and $\btrmms \simreln\ \btrmms'$, where $\sum_i \sum_i \aprobs\PBR{i} \leq \sum_i \bprobs\PBR{i}$.  $\PBR{\bprobs,\abs{\avars}{\btrmms'}}$ has two immediate transitions as follows:
\begin{itemize}
\item Labeled $\ltssignal$ to $\pterm{\sum_i \bprobs\PBR{i}}{\abs{\avar}{\avar}}$.
\item Labeled $\ltsret{\asym }$ to $\subst{\PBR{\bprobs,\btrmms' \A \asub_2}}{\avar}{\asym}$
\end{itemize}
Result follows since $\simfn{\substsimreln} \SEMI\ \leq = \simfn{\substsimreln}$.  

\item $\atrmm_1 = \pterm{\aprob}{\asym \A \atrmms}$.  Let $n = |\atrmms|$.   $\rundef{\asub_1(\asym)}$.  

$\atrmm_1 \A  \asub_1$ has the following transitions.
\begin{itemize}
\item $\ltsfuncall{\asym }{\PBR{i,n}}$ transition to $\pterm{\aprob}{\atrmms\PBR{i} \A \asub_1}$, for $1 \leq i \leq |\atrmms|$; 
\item $\ltsfuncall{\asym }{\PBR{0,n}}$ transition to $\pterm{\aprob}{\Id}$
\end{itemize}
Since $\atrmm_1 \ \simreln\ \atrmm_2$, $\pterm{\aprob}{\asym \A \btrmms} \leq \whnf{\atrmm_2}$, where $|\btrmms| = |\atrmms|=n$ and $\atrmms \simreln\ \btrmms$.  So, $\atrmm_2 \A  \asub_1$ has the following matching transitions.
\begin{itemize}
\item $\ltsfuncall{\asym }{\PBR{i,n}}$ transition to $\pterm{\aprob}{\btrmms\PBR{i}\A \asub_2}$, for $1 \leq i \leq |\btrmms|$; 
\item $\ltsfuncall{\asym }{\PBR{0,n}}$ transition to $\pterm{\aprob}{\Id}$
\end{itemize}
Result follows since $  \pterm{\aprob}{{\asym \A \btrmms}} \A \asub_2  \leq  \atrmm_2\ \A \asub_2 $.  

     \item  $\asub_1(\asym) =\pterm{\bprob} {\bsym \A \btrmms}$.  So, $\atrmm_1 \A \ \asub_1= \pterm{\aprob} {\pterm{\bprob} {\bsym \A \btrmms}} \A  \atrmms \A \asub_1$, and $\rundef{\asub_1\PBR{\bsym}}$.  

Since $\asub_1(\asym) \simreln\ \asub_2(\asym)$, $(\exists \btrmms') \ \btrmms \simreln\ \btrmms'$  and $\pterm{\bprob} {\bsym \A \btrmms'} \leq \whnf{\asub_2(\asym)}$.

Since $\atrmm_1 \simreln\ \atrmm_2$, $\pterm{\aprob}{\asym \A \ctrmms} \in \whnf{\atrmm_2}$, where $\atrmms\ \simreln\ \ctrmms$.  

Transitivity of $\simfnsubst$ (inherited  from $\substsimreln$) on the following subgoals yields the required result.

\begin{enumerate}
\item $ \pterm{\aprob} {\pterm{\bprob} {\bsym \A \btrmms}} \A  \atrmms \A \asub_1 \  \simfnsubst\ \pterm{\aprob} {\pterm{\bprob} {\bsym \A \btrmms}} \A  \atrmms \A \asub_2 $.  

From case(2) since  $\asub_1 \simreln \asub_2, \rundef{\asub_1(\bsym)}$.

\item $\pterm{\aprob} {\pterm{\bprob} {\bsym \A \btrmms}} \A  \atrmms \A \asub_2\ \simfnsubst\ \pterm{\aprob} {\pterm{\bprob} {\bsym \A \btrmms'}} \A \ctrmms \A \asub_2$ 

From case (2), since $\rundef{\asub_1(\bsym)}$, $\btrmms \simreln \btrmms'$ and $\atrmms \simreln \ctrmms$.  

\item $ \pterm{\aprob} {\pterm{\bprob} {\bsym \A \btrmms'}} \A \ctrmms \A \asub_2 \ \leq\ \atrmm_2 \A \asub_2$

\end{enumerate}

\end{enumerate}
\end{proof}

\paragraph{Proof of Lemma~\ref{substsim}}
\begin{proof}
Let   $\atrmm_1 \A \asub_1 \substsimreln\ \atrmm_2 \A \asub_2 $. So, $\atrmm_1 \simreln \atrm_2, \asub_1 \simreln \asub2$.

We need to show that $\whnf{\atrmm_1 \asub_1} \simfnsubst\ \atrmm_2 \A \asub_2$.  Since 
$\simfn{\substsimreln}$ is admissible for $\leq$-chains, it suffices to show this for $\vals{\etrmm_1}$ where $\atrmm_1 \A \ \asub_1 \evals{} \etrmm_1$.  The proof proceeds by induction on the length of $\atrmm_1 \weaktransition{\albl}{} \etrmm_1$.  

The base case is proved in lemma~\ref{basecase}.

If $\atrmm_1 \eval{} \atrmm'_1$, then $\atrmm_1 \A \asub_1 \eval{} \atrmm'_1 \A \asub_1 $.  In this case, result follows from induction hypothesis on $\atrmm'_1 \asub_1, \atrmm_2 \asub_2$, since $ \atrmm'_1\ \bisimreln\ \atrmm_1 \simreln \atrmm_2$, so $\atrmm'_1 \A \asub_1 \substsimreln\ \atrmm_2 \ \A \asub_2$. 

So, we can assume that $\ltsstrongtransition{\atrmm_1 \A \asub_1}{\tau}{}$ is not induced by a reduction $\atrmm_1 \eval{}$ .  

In this case, $\atrmm_1 = \pterm{\aprob}{\asym \A \atrmms_1}$.  $\asym \in \dom{\asub_1}$.  As above we can assume that $\ltsstrongtransition{\atrmm_1 \A \asub_1}{\tau}{}$ is not induced by a reduction $\asub_1(\asym) \eval{}$.   
So,  $\asub_1(\asym) = \PBR{\bprobs,{\abs{\avars}{\atrmms'_1}}}$. Wlog, we assume that $\rundef{\asub_1(\avar)}$.  Then, $\atrmm_1 \A \asub_1 \eval{} \aprob \times \PBR{\bprobs,\atrmms'_1} \A \avars \A \asub'_1$, where
\[
\begin{array}{ lll}
\asub'_1(\bvar) &=& \left\{
                                \begin{array}{l}
                                    \asub_1\PBR{\bvar}, \ \bvar\not=\avar, \defined{\asub_1(\bvar)} \\
                                    \atrmms_1\PBR{i}, \ \bvar = \avars\PBR{i}, i \geq 1
                                \end{array}
                              \right.
\end{array}
\]
and $\avars$ is a vector of variables of same length as $|\atrmms_1|$ such that $\avar = \avars\PBR{1}$. 

Since $\atrmm_1\ \simreln\ \atrmm_2$, there exists $\pterm{\aprob}{\asym \A \atrmms_2} \ \leq\ \whnf{\atrmm_2}$ such that $\atrmms_1 \  \simreln\  \atrmms_2$.   Since $\asub_1(\asym) \ \simreln\ \asub_2(\asym)$,  there exists $\PBR{\bprobs,{\abs{\avars}{\atrmms_2'}}} \leq \whnf{\asub_2(\asym)}$ such that $\PBR{\bprobs,\atrmms'_1} \ \simreln\  \PBR{\bprobs,\atrmms'_2}$. 
Consider  $\aprob \times \PBR{\bprobs,\atrmms'_2} \A \avars \A \asub'_2$, where
\[
\begin{array}{ lll}
\asub'_2(\bvar) &=& \left\{
                                \begin{array}{l}
                                    \asub_2\PBR{\bvar}, \ \bvar\not=\avar, \defined{\asub_1(\bvar)} \\
                                    \atrmms_2\PBR{i}, \ \bvar = \avars\PBR{i}, i \geq 1
                                \end{array}
                              \right.
\end{array}
\]
and $\avars$ is a vector of variables of same length as $|\atrmms|=|\ctrmms|$ such that $\avar = \avars\PBR{1}$. 
Result follows from induction hypothesis on $ \atrmms' \A \asub'_1 , \ctrmms' \A \asub'_2$ and $\leq$ monotonicity of $\simfnsubst$.
\end{proof}